\PassOptionsToPackage{dvipsnames}{xcolor}	
\documentclass[12pt,letter]{article}

\usepackage[T1]{fontenc}
\usepackage{lmodern}
\usepackage{setspace}
\usepackage{etoolbox}
\makeatletter
\patchcmd{\appendix}{\@Alph}{\@Roman}{}{}
\makeatother

\usepackage{amsmath}
\usepackage{amssymb}
\usepackage{amsthm}
\usepackage{mathtools}
\usepackage{mathrsfs}
\usepackage{breqn}
\usepackage{bigints}

\usepackage[margin=1 in]{geometry}
\usepackage{enumerate}
\usepackage{enumitem}
\setlist[enumerate,1]{label=(\arabic*)}   
\usepackage{xcolor}
\usepackage{hyperref}
\usepackage{float}
\usepackage{indentfirst}
\usepackage{caption}
\usepackage{array}

\usepackage{tikz}
\usetikzlibrary{decorations.pathreplacing}
\usepackage{mathtools}
\usetikzlibrary{positioning}
\usepackage{graphicx}

\newcommand{\mcal}{\mathcal}

\renewcommand{\epsilon}{\varepsilon}

\newtheorem{lemma}{Lemma}

\newtheorem{proposition}{Proposition}

\DeclareMathOperator*{\argmax}{\arg\!\max}

\DeclareTextFontCommand{\emph}{\slshape}

\usepackage{natbib}
\nocite{*}

\title{Quantity Limits on Addictive Goods}
\author{Eric Gao\thanks{Department of Economics, Stanford University. egao2@stanford.edu.}}
\date{\today}

\begin{document}
	
\maketitle

\begin{abstract}
Addiction is a major societal issue leading to billions in healthcare losses per year. Policy makers often introduce ad hoc quantity limits---limits on the consumption or possession of a substance---something which current economic models of addiction have failed to address. This paper enriches \cite{bernheim_2004_addiction}'s model of addiction driven by cue-triggered decisions by incorporating endogenous choice of \textit{how much} of the addictive good to consume, instead of just whether or not consumption happens. Stricter quality limits improve welfare as long as they do not preclude the myopically optimal level of consumption.
\end{abstract}
\noindent \bigskip \bigskip \textbf{JEL Codes}: D01, D11, I18.

\onehalfspacing

\newpage

\section{Introduction and Motivation}

Misuse of addictive substances costs the United States billions in healthcare spending (\cite{peterson_2021_assessment}) and over a hundred thousand deaths every year (\cite{_2023_drug}). While there has been much work in economics dedicated to analyzing this issue and evaluating policy proposals, such inquires have focused on full criminalization versus legalization and the impact of taxes or subsidies. In the real world, however, policies often feature limits on the quantity an individual can possess: California allows possession of up to 28.5 grams of recreational cannabis, while most states where cannabis is legal have possession limits of around 1-3 ounces.\footnote{See \cite{pacula_2021_current} for a more detailed breakdown of policies by area.} Similarly, California recently passed laws prohibiting bartenders from serving alcohol ``to anyone who is obviously intoxicated'' (\cite{intoxication}). Not only that, but the California law also prohibits alcohol from being served to anyone ``who has lost control over his or her drinking''.\footnote{This can be thought of as a ``higher-order'' or ``proactive'' quantity limit, but falls slightly beyond the scope of this paper.} 

What are the behavioral implications of such limits, and how does that spill over to impact public welfare? While limits on the quantity of addictive goods someone can possess  restricts how much they consume at once, it can also create a perception that addiction ``isn't too bad'', since the worst that can happen is consumption at the limit. As such, quantity limits may may counter-productively increase overall usage by increasing individuals' rationally optimal quantity consumed (while decreasing compulsive consumption). For instance, consider a recovering alcoholic who doesn't go to bars due to fear of losing control, drinking too much, and ending up in the hospital. If a law were passed preventing bartenders from serving drunk customers, they may now venture into bars, thinking that even if \textit{they} lose control, the bartender will exercise restraint on their behalf.

With this in mind, suppliers of addictive goods may find it beneficial to impose their own limits on how much a customer can purchase at once to increase the number of customers (at the cost of each individual customer consuming less if the limit is binding). One example of this could be table limits in casinos restricting how much an individual can bet at a time (but the common explanation of why table limit exists is to protect the casino itself from risk). If quantity limits both increase welfare and are profitable for suppliers (perhaps for goods where continuous low usage has negligible harms, but overuse leads to enormous costs), supply-side restrictions at the store-by-store level can be a useful policy tool. The remainder of this paper formalizes these ideas and proposes an experimental evaluation of quantity limits.

\subsection{Related Literature}

This work is most closely related to \cite{bernheim_2004_addiction}. Their seminal paper is the first to incorporate a psychological foundation of rational consumption of addictive goods. Addictive substances are addictive because they directly trigger a response in the brain's mesolimbic dopamine system which regulates feelings of happiness. When certain environmental cues are presented in conjunction with stimulation of the mesolimbic dopamine system, laboratory experiments have shown that the brain learns to associate dopamine with the cues themselves opposed to eventual rewards. As such, \cite{bernheim_2004_addiction} model the process of compulsive consumption through individuals entering a ``hot'' state after being triggered by some environmental. 

In addition to introducing the novel model, the authors derive comparative statics on how lifestyle choices that bring different levels of utility, but also different levels of exposures to cues, change as model primitives change. Instead of having consumption of the addictive good being a binary choice, the individual's action space is enriched to being a continuous interval. Being less interested in the dynamics associated with lifestyle choice, my model only has two possible lifestyle activities opposed to three. Unfortunately, there is not much work utilizing models with cue-triggered decision-making beyond \cite{bernheim_2004_addiction}.

Alternative models of rational addiction exist. \cite{gul_2007_harmful} considers a similar model where decision-makers sometimes make compulsive choices. However, ``compulsive consumption'' in their model comes from goods having some temptation utility opposed to triggered decisions due to environmental cues. The main drawback of such a model is that it becomes difficult to conduct positive welfare analysis: when preferences are not consistent, should temptation utility be considered part of the decision-maker's well-being? 

Another class of rational addiction models stems from \cite{becker_1988_a}; recently, \cite{allcott_2022_digital} studies overuse of social media as addiction to digital goods while \cite{hussam_2022_rational} investigates habit formation behind regular hand-washing using a similar model. Addiction in these models comes from two sources: temptation utility as discussed previously, and higher marginal utility from consuming the good as prior consumption increases. However, someone with a higher history of past consumption might not necessarily enjoy consumption more today (except through being more tempted). Furthermore, this paper does not rule out interdependence between current utility from consumption and past consumption; higher sensitivity to quantity consumed as prior consumption increases can be viewed as a special case of my model. While \cite{allcott_2022_digital} considers the impact of a quantity limit experimentally, the limit is self-prescribed (participants could choose how much screen time to give themselves). This paper focuses on the impact of externally-imposed limits.   

Other papers such as \cite{irvine_2011_toxic} study the welfare implications of smoking and bans on smoking. Their theoretical model of how smoking influences smoker utility can help inform the shape of utility functions in my model while their empirical estimations can help inform population parameters that can come into play when looking at aggregate welfare or total profits of firms selling addictive goods. \cite{decicca_2022_the} provides an overview of different policies that seek to regulate tobacco, while \cite{levy_2018_tobacco} analyzes how regulatory agencies have analyzed consumer surplus under different policies. \cite{odermatt_2015_smoking} empirically analyzes the impact of taxes on smoking, finding that higher prices reduce life satisfaction of habitual smokers but increases the well-being of smokers who are trying to quit.

Many reviews of the different behavioral components that influence gambling decisions exist, such as \cite{gainsbury_2018_behavioral}. Similarly, \cite{stetzka_2021_how} evaluates how rational gambling behavior is, evaluating multiple different potential explanations of why individuals choose to gamble. \cite{cameron_2022_an} develops a two-period model of gambling behavior, where gambling brings positive utility in the first period but cravings for gambling take over and gambling becomes detrimental in the second period. In terms of policies aimed at mitigating the harms of compulsive gambling, \cite{broda_2008_virtual} analyzes the impact of self-imposed betting limits on an online gambling site, finding that (1) most users do not violate their self-imposed limits, and (2) those that do go above self-imposed limits do not suffer from poor outcomes and bet rationally. Similarly, \cite{badji_2023_economic} investigates the welfare effects of increased gambling availability and find that financial and mental well-being decrease as proximity to gambling venues increase.

\section{Model}

The model largely follows \cite{bernheim_2004_addiction}. There is a decision-maker (DM) that operates in an infinite-horizon discrete time setting who discounts the future at rate $\delta$. At every period the DM observes their addictive state $s \in \mathbf{R}$ and choose a lifestyle activity from $A = \{E, R\}$ where $E$ denotes exposure to the addictive good and $R$ denotes rehabilitation. If $R$ is chosen, the DM needs to pay a monetary cost of $r(s)$ for some continuous $r$, no amount of the addictive good is consumed, and the period ends. If $E$ is chosen, the DM can then choose to purchase some quantity $q \in [0, Q]$ of the addictive good at a price $p$ per unit, spending a total of $pq$. 

With some probability depending on the history of past consumption, the DM may enter a triggered state if $E$ is chosen and compulsively consumes $q = Q$ of the addictive good. To model this, the DM's past consumption determines their addictive state $s$. At state $s$, consumption of $q$ leads to state $\lambda s + q$ in the next period, for some $\lambda \in (0,1)$. There are two consequences of this setup: 
\begin{enumerate}
    \item Addictive state is upper bound by $\overline{S} = Q/(1-\lambda)$: If $s \leq Q/(1-\lambda)$, then consumption of $Q$ leads to an addictive state of
    $$\lambda s + Q \leq \lambda \frac{Q}{1-\lambda} + Q = \frac{(\lambda+(1-\lambda))Q}{1-\lambda} = \frac{Q}{1-\lambda}.$$
    \item If the DM ever leaves the state $s = 0$, they can never return: At any $s > 0$, we have that $\lambda s + q \geq \lambda s > 0$.
\end{enumerate}
With (1) in mind, given some arbitrary upper bound on upper bounds $\overline{Q}$, let the set of possible state-quantity limit pairs be
$$\mcal D = \{(s,Q): 0 \leq Q \leq \overline{Q}, 0 \leq s \leq Q/(1- \lambda)\}.$$
For finite $\overline{Q}$, we have that $\mcal D$ is a compact subset of $\mathbb{R}^2$. For all intents and purposes, $\overline{Q}$ can be taken to be large enough to encompass the relevant region of quantity limits. Incorporating the quantity limit into the state space and restricting attention to a a compact subset will be useful for applying dynamic programming results when deriving comparative statics.

The DM's perceived attractiveness of consuming the addictive good additionally depends on some cue (such as advertising) $\omega \in \Omega$ drawn by nature according to distribution $\mu$. Then, $M: A \times S \times \Omega \to \mathbb{R}$ measures the perceived attractiveness of consuming the addictive good. We assume that:
\begin{enumerate}
    \item No attractiveness in rehabilitation: $M(R,s, \omega) = 0$ for all $s, \omega$.
    \item No attractiveness if no history: $M(a, 0, \omega) = 0$ for all $a, \omega$. 
    \item Higher prior use leads to additional perceived attractiveness: $M(E, s, \omega)$ is increasing in $s$ for $s \in [0, \overline{S}]$.
\end{enumerate}
The DM enters the triggered state if $M(a, s, \omega)$ is larger than some threshold $M^T$. Let $\mcal C(a, s) = \{\omega \in \Omega: M(a, s, \omega) \geq M^T\}$ be the set of cues that leads to the triggered state and $\mcal P(a,s) = \mu(\mcal C(a,s))$ be the induced probability of entering into the triggered state. A direct consequence of our three assumptions on $M$ is the following: 
\begin{lemma}[Properties of $\mcal P$]
    The probability of entering the triggered state satisfies: 
    \begin{enumerate}
		\item No trigger in rehabilitation: $\mcal P(R,s) = 0$ for all $s$.
		\item No trigger if no history: $\mcal P(a,0) = 0$ for all $a$.
		\item Higher prior use leads to higher chance of trigger:  $\mcal P(E,s)$ is increasing in $s$ for $s \in [0, \overline{S}]$.
    \end{enumerate}
\end{lemma}

In each period, the DM receives income $y(s)$ that depends continuously on their addictive state. Their payoff in each state is determined by their addictive state, income, quantity of addictive good purchased, and activity captured in the function $u: S \times \mathbb{R} \times [0, Q] \times A \to \mathbb{R}$. In a given period, if the DM chooses rehabilitation, their payoff is
$$u(s, y(s) - r(s), 0, R).$$
If the DM chooses exposure and consumption of $q$ units of the addictive good, with probability $1-\mcal P(E,s)$ they do not enter the triggered state and collect a payoff of
$$u(s, y(s) - pq, q, E)$$
and with probability $\mcal P(E,s)$ they do enter the triggered state and collect a payoff of 
$$u(s, y(s) - pQ, Q, E).$$
Let $V: \mcal D \to \mathbb{R}$ be the DM's value function of their stochastic programming problem given their current addictive state is $s$ and the quantity limit $Q$ (recall that $\mcal D$ is the relevant domain of possible additive state-quantity limit pairs). With this, we can define 
\begin{equation*}
    \begin{split}
    q^*(s,Q) = \argmax_{q \in [0, Q]} & \Big\{(1- \mcal P(E,s)) [u(s, y(s) - pq, q, E) + \delta V(\lambda s+q, Q)] \\
    & + \mcal P(E,s)[u(s, y(s) - pQ, Q, E) + \delta V(\lambda s+Q, Q)] \Big\}.
    \end{split}
\end{equation*}
to be the DM's optimal level of consumption of the addictive good conditional on choosing exposure and at addictive state $s$. As 
$$1- \mcal P(E,s) \text{ and }\mcal P(E,s)[u(s, y(s) - pQ, Q, E) + \delta V(\lambda s + Q, Q)]$$
do not depend on $q$, we can equivalently define
$$q^*(s,Q) = \argmax_{q \in [0, Q]} u(s, y(s) - pq, q, E) + \delta V(\lambda s+q, Q).$$
If the DM decides to choose an activity of $E$, they will always choose $q^*(s)$ as their level of consumption of the addictive good. Thus, the value function solves
\begin{equation*}
    \begin{split}
        V(s, Q) = \max \Bigg\{ &(1-\mcal P(E,s))[u(s,y(s)-pq^*(s),q^*(s), E) + \delta V(\lambda s + q^*(s), Q)] \\ 
        &+\mcal P(E,s)[u(s,y(s)-pQ,Q, E) + \delta V(\lambda s + Q, Q)], \\
        &u(s,y(s),0, R) + \delta V(\lambda s, Q) \Bigg\}.
    \end{split}
\end{equation*}

A crucial note about the model is that addictive dynamics only \textit{necessarily} depend on the probability the DM compulsively consumes. In particular, this means that we do not require DMs who have a higher period of consumption to be more sensitive to the quantity of the addictive good consumed. There are many cases of individuals who want to stop using an addictive substance and recognize that usage is harmful, yet are unable to stop. However, our model can incorporate this case: If more addicted DM's are more sensitive to the quantity of the good consumed, that can be modeled by assuming that
$$\left|\frac{\partial}{\partial q}u(s, y(s) - pq, q, E)\right|$$
is increasing in $s$. 

Going forward, let
\begin{equation*}
    \begin{split}
        \mcal P(s) &= \mcal P(E,s); \\
        u(s, q) &= u(s,y(s)-pq,q, E); \\
        u(s, Q) &= u(s,y(s)-pQ,Q, E); \\
        u(s, R) &= u(s,y(s)-r(s),0, R)
    \end{split}
\end{equation*}
to simplify notation whenever there is no ambiguity.

\section{Results}

The main result is that if the DM dislikes being at a higher addictive state, then tightening the quantity limit increases their utility, as long as it never infringes on myopically optimal consumption. This extends the standard behavioral intuition that helping the DM commit to not taking actions they do not want to take may make them better off into this setting. We will work towards the following formal result: 

\begin{proposition}[Properties of the Value Function]\label{value}
The value function satisfies the following:
    \begin{enumerate}
        \item If $u$ is continuous in $s$ and $q$, then $V$ is continuous in $s$ and $Q$;
        \item If $u$ is decreasing in $s$, then $V$ is decreasing in $s$ for all $Q$;
        \item Suppose that in addition to (2), there exists $\overline{q}$ such that $u(s, y(s)-pq, q,E)$ is decreasing in $q$ for all $q \geq \overline{q}$. Then, $V$ is increasing in $Q$ for all $s$ and $Q \geq \overline{q}$;
        \item If $u$ is (strictly) concave in $s$, then $V$ is (strictly) concave in $s$.
    \end{enumerate}
\end{proposition}

All three results follow from standard dynamic programming techniques. Given any function $W: \mcal D \to \mathbf{R}$, define the ``value iteration'' functional as follows:
\begin{equation*}
    \begin{split}
        (TW)(s, Q) = \max \Bigg\{ &\max_{q \in [0,Q]}\Big\{(1-\mcal P(s))[u(s,q) + \delta W(\lambda s + q, Q)] \\ 
        &\hspace{3em}+\mcal P(s)[u(s,Q) + \delta W(\lambda s + Q, Q)]\Big\}, \\
        &u(s, R) + \delta W(\lambda s, Q) \Bigg\}.
    \end{split}
\end{equation*}
Then, the value function $V$ is a fixed point of $T$, so $V = TV$. Value iteration preserves each of the three properties, so any fixed point must also have the three properties.

From a policy perspective, this result establishes that if society is confident that no rational individual would optimally consume above some quantity of good, then a quantity limit at that level can only increase welfare. A shortcoming of the current analysis is that we have only considered a single decision-maker: What if there were many individuals with heterogeneous preferences over the addictive good? If there were a (potentially non-binding, for some individuals who do not like the addictive good very much) uniform bound on the rationally optimal quantity consumed over all individuals, such a quantity limit would still be welfare-improving.

\subsection{Quantity Consumed}

When the DM decides how much of the addictive substance to consume, they trade off between three competing forces: Less income today due to more spending on the addictive substance, increased experiential utility from consumption itself, and lowered utility tomorrow from being at a higher addictive state. Changing the consumption limit leaves the first two effects unchanged, but can shift continuation payoffs. While it is intuitively tempting to say that a looser quantity limit magnifies the harms of addiction, this holds regardless of the addictive state. Instead, the DM cares about how changes in the consumption limit changes the \textit{marginal} harms from a higher addictive state.

\begin{proposition}[Quantity Consumed as Limit Changes]\label{quantity_limit}
    If $Q$ is greater than the myopic optimal level of consumption at state $s$, then $q^*(s,Q)$ is (weakly) increasing as $Q$ decreases if $V$ has strictly \textit{decreasing} differences in $s, Q$. In particular, a sufficient condition is for
    $$\frac{\partial^2}{\partial q \partial Q} V(\lambda s + q, Q) < 0.$$
\end{proposition}

This corresponds to the intuition that external commitment from the quantity limit is a substitute for internal commitment from choosing a lower level of rational consumption. When the quantity limit is sufficiently extreme, we are able to offer a sharper characterization. In this case, the restrictiveness of external commitment completely replaces the need for any internal commitment. 

\begin{proposition}[Consistent Consumption]\label{extreme}
    For every model specification, there exists $\underline{Q} > 0$ such that the DM chooses exposure at every addictive state.
\end{proposition}

Unfortunately, comparative statics as to how quantity consumed changes as addictive state changes are harder to pin down. In particular, the curvature of the DM's value function with respect to addictive state influences their choice of how much to consume. As such, it is difficult to pin down how consumption changes as addictive state changes purely in terms of model primitives.

\subsection{Numerical Solution}

Consider the following parameterization of the model. Take $\lambda = 0.8$. Let the DM's within-period utility from choosing exposure and consuming $q$ be
$$u(s,q) = \sqrt{100-s} - q +q(10-q).$$
This could potentially represent a DM with utility that is quasilinar in income and addictive state, enjoys experiential utility of $\sqrt{100-s}$ from being at a given addictive state, pays a price of one for the addictive good, and consumption utility of $q(10-q).$ If continuation payoffs are inconsequential, then the in-period optimum is to consume $q^* = 4.5$ regardless of the addictive state. The probability the DM enters the hot state is $\mcal P(s) = \sqrt{s}/15$.\footnote{A range of denominators were tried. If the probability is too low, the DM always consumes. If the probability is too high, the DM always chooses rehabilitation. In general, these functional forms were chosen to provide a clean visualization of general results.}

The simplified model is solved using value iteration in Matlab. Value functions are as follows, and confirms the result that tightening the quantity limit improves the DM's utility, as long as it does not restrict the myopically optimal level of consumption:

\begin{figure}[h]
    \centering
    \includegraphics[width = 0.8\textwidth]{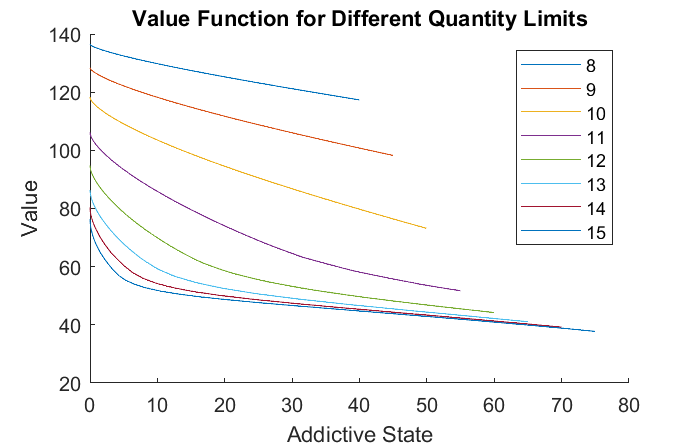}
    \caption{Value Function at Different Quantity Limits}
\end{figure}

Next, we plot the DM's policy function, where a quantity choice of $-1$ corresponds to rehabilitation. Trembles in the DM's consumption choice are most likely due to the solution being numeric, as opposed to being a consequence of the model. 

\begin{figure}[h!]
    \centering
    \includegraphics[width = 0.8\textwidth]{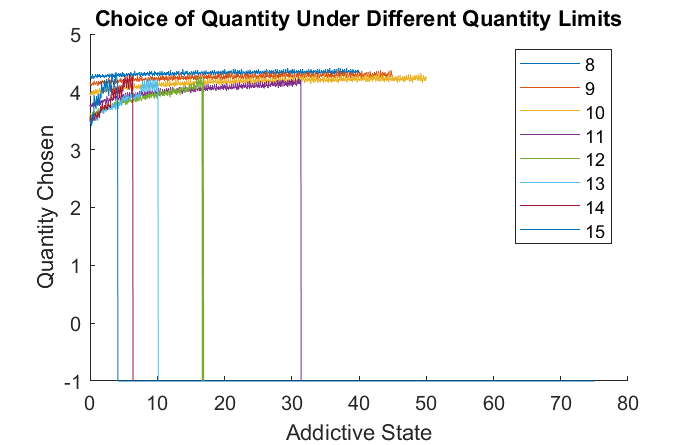}
    \caption{Policy at Different Quantity Limits}
\end{figure}

When the quantity limit decreases from $10$ to $9$ to $8$, quantity consumed at every state increases. However, that does not hold at other quantity levels: When the quantity limit decreases in the range of $15$ to $11$, the set of states in which the DM chooses to consume increases, yet the quantity consumed decreases. 

In this case, there is a ``last hurrah'' effect: \cite{valenzky_2023_the} describes how often times, drinkers think that it is fine to let loose one last time before entering rehabilitation. If the problem is going to be solved tomorrow anyways, then why not drink again today? In this particular model specification, this effect dominates the disutility from a higher addictive state, leading to consumption close to the myopically optimal level right before rehabilitation. 

\section{Conclusion}

Consumption of addictive goods is a pressing concern that the economics literature does not have a good answer for. In particular, the impact on limits of how much of an addictive substance one can consume at a time is severely understudied, yet policy regulations in the real world often takes this form. This paper takes a first step at analyzing such quantity limits. 

Our main result shows that if a quantity limit does not infringe on anyone's immediate optimal level of consumption, it can only increase welfare. For instance, if the quantity limit is set at the quantity where individuals start overdosing, society can be quite confident that no one's optimal level of consumption is above that.

Further research can extend results in several directions. First is a sharper characterization of what the optimal quantity limit is. In general, different quantity limits can have differing effects on people in different addictive states: Looser limits might be more beneficial if there is a large mass of people at low addictive states, while sharper limits might be more beneficial if there is a large mass of people at high addictive states. Finding ways to aggregate welfare over a while population would produce interesting (and more applicable) results. \cite{bernheim_2004_addiction} pilot this analysis, considering a large economy of many short-lived DM's that transition through different addictive states. Second, if it is indeed the case that sharper quantity limits lead to increased rational consumption, quantity limits may also help sellers of such goods. Characterizing when a sharper quantity limit both increases welfare and producer surplus would be a useful policy tool. In these cases, forcing sellers to commit to only supplying a certain amount at most would benefit both sides of the market, providing another implementation mechanism for quantity limits. Finally, compulsive consumption may not respect the quantity limit (for instance, people still get arrested for public intoxication or possessing more than the legal amount of marijuana). To account for this, instead of consuming $\overline{Q}$ in the triggered state, DMs may instead consume some $\hat{q}(s)$, which could potentially be above $\overline{Q}$. For instance, someone who has a long history of consumption may build up a tolerance for the addictive good, leading to larger and larger amounts needed to achieve the same previous levels of happiness. Enriching the model in these dimensions would help garner a better understanding of how addiction works, and in turn, spur more effective policy responses.

\newpage

\bibliography{cites.bib}

\newpage

\appendix
\section*{Appendix: Omitted Proofs}

\subsection*{PROOF OF PROPOSITION \ref{value}}

\begin{proof}
We will prove the proposition through a series of Lemmas. Our first result establishes that value iteration maps continuous functions to continuous functions:

\begin{lemma}\label{cont}
    If $W$ is continuous in $s$ and $Q$, then so is $TW$.
\end{lemma}

\begin{proof}
    The maximum of continuous functions is continuous and     
    $$u(s, R) + \delta W(\lambda s, Q)$$
    is clearly continuous by preceding assumptions, so it suffices to show that 
    \begin{equation*}
        \begin{split}
        &\max_{q \in [0,Q]}\Big\{(1-\mcal P(s))[u(s,q) + \delta W(\lambda s + q, Q)] +\mcal P(s)[u(s,Q) + \delta W(\lambda s + Q, Q)]\Big\}
        \end{split}
    \end{equation*}
    is continuous. Let this function be $h(s, Q)$.

    Let $X = \mcal D$ and $Y = \mathbb{R}$. Let $\Gamma: X \to Y$ be the correspondence defined by $\Gamma(s, Q) = [0, Q]$. Then, $\Gamma$ is compact-valued and continuous. Let
    \begin{equation*}
        \begin{split}
            f(s, Q, q) =& (1-\mcal P(s))[u(s,q) + \delta W(\lambda s + q, Q)] + \mcal P(s)[u(s, Q) + \delta W(\lambda s + Q, Q)].
        \end{split}
    \end{equation*}
    By Berge’s Theorem of the Maximum, we have that
    $$h(s, Q) = \max_{q \in [0, Q]} f(s, Q, q) = \max_{q \in \Gamma(s, Q)} f(s, Q, q)$$
    is continuous as desired.
\end{proof}

Our next result establishes that if utility is decreasing in addictive state at a proposed value function, then utility is still decreasing in state after applying value iteration to it:

\begin{lemma}\label{noninc}
    If $W(s, Q)$ is non-increasing in $s$ for any $Q$, then $TW$ is non-increasing in $s$ for any $Q$ as well.
\end{lemma}

\begin{proof}
    Suppose $W$ is non-increasing in $s$. Then, 
    $$u(s,R) + \delta W(\lambda s, Q)$$
    is non-increasing in $s$ so it suffices to show that 
    \begin{equation*}
        \begin{split}
            &\max_{q \in [0,Q]}\Big\{(1-\mcal P(s))[u(s,q) + \delta W(\lambda s + q, Q)] +\mcal P(s)[u(s,Q) + \delta W(\lambda s + Q, Q)]\Big\}
        \end{split}
    \end{equation*}
    is non-increasing in $s$. The above is equivalent to 
    \begin{equation*}
        \begin{split}
            &(1-\mcal P(s))\max_{q \in [0,Q]}\Big\{(u(s,q) + \delta W(\lambda s + q, Q)\Big\} +\mcal P(s)[u(s,Q) + \delta W(\lambda s + Q, Q)].
        \end{split}
    \end{equation*}
    As $q = Q$ is in the choice set, 
    $$\max_{q \in [0,Q]}\Big\{(u(s,q) + \delta W(\lambda s + q, Q)\Big\} \geq u(s,Q) + \delta W(\lambda s + Q, Q).$$
    Thus, if $s' > s$, we have that $\mcal P(E,s') \geq \mcal P(E, s)$ so
    \begin{equation*}
        \begin{split}
            &(1-\mcal P(s))\max_{q \in [0,Q]}\Big\{(u(s,q) + \delta W(\lambda s + q, Q)\Big\} +\mcal P(s)[u(s,Q) + \delta W(\lambda s + Q, Q, Q)] \\
            \geq &(1-\mcal P(s'))\max_{q \in [0,Q]}\Big\{(u(s,q) + \delta W(\lambda s + q, Q)\Big\} +\mcal P(s')[u(s,Q) + \delta W(\lambda s + Q, Q)]. \\
        \end{split}
    \end{equation*}
    As $u$ and $W$ are non-increasing in $s$, we also have that 
    \begin{equation*}
        \begin{split}
            &(1-\mcal P(s'))\max_{q \in [0,Q]}\Big\{(u(s,q) + \delta W(\lambda s + q, Q)\Big\} +\mcal P(s')[u(s,Q) + \delta W(\lambda s + Q, Q)] \\
            \geq &(1-\mcal P(E,s'))\max_{q \in [0,Q]}\Big\{(u(s',q) + \delta W(\lambda s' + q, Q)\Big\} +\mcal P(s')[u(s', Q) + \delta W(\lambda s' + Q, Q)]. \\
        \end{split}
    \end{equation*}
\end{proof}

A similar result holds to show (3): If a proposed value function is increasing in $Q$ for $Q \geq \overline{q}$, then performing value iteration preserves that property:

\begin{lemma}\label{limit_good}
    Suppose there exists $\overline{q}$ such that $u(s, y(s)-pq, q,E)$ is decreasing in $q$ for all $q \geq \overline{q}$ and $s$. If $W(s,Q)$ is non-increasing in $s$ for all $Q$ and non-increasing in $q$ for all $q \geq \overline{q}$ and $s$, then $(TW)(s,Q)$ is non-increasing in $s$ for all $Q$ and non-increasing in $Q$ for all $Q \geq \overline{q}$ and $s$.
\end{lemma}

\begin{proof}
    Lemma \ref{noninc} gives that $(TW)(s,Q)$ is non-increasing in $s$ for all $Q$. 
    
    By assumption of $W$ being non-increasing in $Q$ for $Q \geq \underline{Q}$ and $u(s,y(s),0, R)$ being independent of $Q$, we have that 
    $$u(s, R) + \delta W(\lambda s, Q)$$
    is non-increasing in $Q$ for $Q \geq \underline{Q}$. 

    Similarly, by assumption of $W$ being non-increasing in $Q$ for $Q \geq \underline{Q}$, $W$ being non-increasing in $s$, and $Q \geq \underline{Q}$ combined with $u(s,q)$ being non-increasing in $q$ for $q \geq \underline{Q}$, we have that
    $$u(s,Q) + \delta W(\lambda s+Q, Q)$$
    is non-increasing in $Q$ as well.

    As such, it suffices to show that 
    $$\max_{q \in [0,Q]}\Big\{(u(s,q) + \delta W(\lambda s + q, Q)\Big\}$$
    is non-increasing in $Q$. Fix $s$ and $\underline{Q} \leq Q < Q'$ with $(s,Q), (s,Q') \in \mcal D$ noting that if no such $Q, Q'$ exist, then the lemma is vacuously true. Recall that
    $$q^*(s, Q) = \argmax_{q \in [0,Q]}\Big\{(u(s,q) + \delta W(\lambda s + q, Q)\Big\}.$$
    For any $Q \geq \underline Q$, we have that $q^*(s, Q) \leq \underline Q$ since both $u$ and $W$ are decreasing past $\underline Q$. Then,
    \[
    \begin{array}{>{\displaystyle}r>{\displaystyle}l>{\displaystyle}l}
        &\max_{q \in [0,Q]}\Big\{(u(s,q) + \delta W(\lambda s + q, Q)\Big\} &  \\
        =& u(s,q^*(s, Q) + \delta W(\lambda s + q^*(s, Q), Q) & \text{by definition of } q^*\\
        \geq& u(s,q^*(s, Q') + \delta W(\lambda s + q^*(s, Q'), Q) & q^*(s, Q') \in [0, \underline Q] \subseteq [0, Q]\\
        \geq& u(s,q^*(s, Q') + \delta W(\lambda s + q^*(s, Q'), Q') & W \text{ non-increasing in }Q \\
        =& \max_{q \in [0,Q']}\Big\{(u(s,q) + \delta W(\lambda s + q, Q')\Big\} & 
    \end{array}
    \]
    as desired.
\end{proof}

\begin{lemma}\label{concave}
    Suppose $u$ is (strictly) concave in $s, q$ and $W$ is concave. Then, $TW$ is (strictly) concave in $s$.
\end{lemma}

\begin{proof}
    We have that for any $Q$,
    \begin{equation*}
        \begin{split}
            & TW(\gamma s + (1-\gamma)s', Q) \\
            = & (1-\mcal P(\gamma s + (1-\gamma)s')) \Big[ u(\gamma s + (1-\gamma)s', q^*(\gamma s + (1-\gamma)s'), Q)) \\
            & + \delta W(\lambda(\gamma s + (1-\gamma)s') + q^*(\gamma s + (1-\gamma)s'), Q)\Big] \\
            & + \mcal P(\gamma s + (1-\gamma)s') \left[ u(\gamma s + (1-\gamma)s', Q) + \delta W(\lambda(\gamma s + (1-\gamma)s')+Q, Q)\right] \\
            \geq & (1-\mcal P(\gamma s + (1-\gamma)s')) \Big[ u(\gamma s + (1-\gamma)s', \gamma q^*(s, Q) + (1-\gamma) q^*(s', Q)) \\
            & + \delta W(\lambda(\gamma s + (1-\gamma)s') + \gamma q^*(s, Q) + (1-\gamma) q^*(s', Q), Q)\Big] \\
            & + \mcal P(\gamma s + (1-\gamma)s') \left[ u(\gamma s + (1-\gamma)s', Q) + \delta W(\lambda(\gamma s + (1-\gamma)s')+\gamma Q + (1-\gamma) Q, Q)\right] \\
            \geq & (1-\mcal P(\gamma s + (1-\gamma)s')) \Big[ \gamma u(s, q^*(s, Q)) + (1-\gamma) u(s', q^*(s', Q)) \\
            & + \delta (\gamma W(\lambda s + q^*(s, Q), Q) + (1-\gamma)W(\lambda s' + q^*(s', Q), Q))\Big] \\
            & + \mcal P(\gamma s + (1-\gamma)s') \left[ \gamma u( s, Q) +  (1-\gamma)u(s', Q) + \delta(\gamma W(\lambda s + Q, Q) + (1-\gamma)W(\lambda s' + Q, Q))\right] \\
            = & \gamma \Big\{(1-\mcal P(\gamma s + (1-\gamma)s'))[u(s, q^*(s, Q)) + \delta W(\lambda s + q^*(s, Q), Q)]  \\
            & + \mcal P(\gamma s + (1-\gamma)s')[u(s, Q) + \delta W(\lambda s + Q, Q)]\Big\} \\
            & + (1-\gamma) \Big\{(1-\mcal P(\gamma s + (1-\gamma)s'))[u(s', q^*(s', Q)) + \delta W(\lambda s' + q^*(s', Q), Q)]  \\
            & + \mcal P(\gamma s' + (1-\gamma)s')[u(s', Q) + \delta W(\lambda s' + Q, Q)]\Big\}
        \end{split}
    \end{equation*}
\end{proof}

Finally, we show that underlying space of possible value functions is complete and value iteration is a contraction mapping to guarantee that value iteration converges to a unique fixed point:

\begin{lemma}\label{complete}
    The space of $L_\infty$-uniformly bounded functions $W: \mcal D \to \mathbb{R}$ that have any subset of the following properties:
    \begin{enumerate}
        \item non-increasing in $s$;
        \item non-increasing in $Q$ for all $Q \geq \overline{q}$ and all $s$;
        \item concave;
    \end{enumerate}
    is complete.
\end{lemma}

\begin{proof}
    The space of functions $W: \mcal D \to \mathbb{R}$ that are bounded is complete in the sup norm. Then, the set of functions that satisfy condition (1), (2), or (3) is closed. Intersections of closed sets are closed, which implies that the set of functions that satisfy any subset of these conditions is closed. Finally, closed subspaces of complete spaces are complete. 
\end{proof}

\begin{lemma}
    Value iteration is a contraction mapping.
\end{lemma}

\begin{proof}
    By Blackwell’s Sufficient Conditions for a functional to be a contraction mapping, it suffice to show that 
    \begin{enumerate}
        \item $W \leq U$ implies $TW(s,Q) \leq TU(s,Q)$;
        \item There exists $\beta \in (0,1)$ such that $(T[W+a])(s,Q) \leq TW(s,Q) + \beta a$ for $a \geq 0$.
    \end{enumerate}
    The first point is immediate. Then,
    \begin{equation*}
    \begin{split}
        (T[W+a])(s,Q) = \max \Bigg\{ &\max_{q \in [0,Q]}\Big\{(1-\mcal P(s))[u(s,q) + \delta (W(\lambda s + q, Q)+a)] \\ 
        &\hspace{3em}+\mcal P(s)[u(s,Q) + \delta W(\lambda s + Q, Q)]\Big\}, \\
        &u(s, R) + \delta (W(\lambda s, Q)+a) \Bigg\} \\
        = \max \Bigg\{ &\max_{q \in [0,Q]}\Big\{(1-\mcal P(s))[u(s,q) + \delta W(\lambda s + q, Q)] \\ 
        &\hspace{3em}+\mcal P(s)[u(s,Q) + \delta W(\lambda s + Q, Q)]\Big\}, \\
        &u(s, R) + \delta W(\lambda s, Q) \Bigg\} + \delta a \\
        & \hspace{-4em} = TW (s,Q) + \delta a
    \end{split}
    \end{equation*}
    where $\delta a$ is constant and hence affects neither maximization problem. Thus, taking $\beta = \delta$ suffices.
\end{proof}

To prove the Theorem, we can start by taking $W(s,Q) = 0$ for all $s, Q$. This is continuous in $s,Q$, non-increasing in $s$ for any $Q$, and non-increasing in $q$ for any $s$. Then, $TW$ satisfies all these properties as well. Finally, the sequence $\{T^kW\}_{k=1}^\infty$ converges as the space of all functions is complete and value iteration is a contraction mapping. Furthermore, it converges to a unique value function $V$ that inherits all desired properties. 

\end{proof}

\subsection*{PROOF OF PROPOSITION \ref{quantity_limit}}
\begin{proof}
    Recall that 
    \begin{equation}\label{maximand}
        q^*(s,Q) = \argmax_{q \in [0, Q]} u(s, q) + \delta V(\lambda s+q, Q).
    \end{equation}
    We will show that strictly decreasing differences of $V$ in $s, Q$ is a sufficient condition for the overall maximand to have strictly decreasing differences in $q, Q$. Then, by the Monotone Selection Theorem of \cite{milgrom_1994_monotone}, the DM's optimal selection of $q$ is decreasing as $Q$ increases.

    By strict decreasing differences in $V$, for any $\lambda, s, q, q', Q, Q'$ with $(s, Q), (s', Q'), (\lambda s + q, Q), (\lambda s + q', Q') \in \mcal D$, we have that
    $$V(\lambda s+q',Q') - V(\lambda s+q,Q') > V(\lambda s+q',Q) - V(\lambda s+q,Q).$$
    Adding
    $$u(s, q')-u(s, q)$$
    to both sides gives that
    \begin{equation*}
        \begin{split}
            &\big(V(\lambda s+q',Q')+u(s, q')\big)-\big(V(\lambda s+q',Q')+u(s, q)\big) \\
            >&\big(V(\lambda s+q',Q)+u(s, q)\big)-\big(V(\lambda s+q',Q)+u(s, q')\big)
        \end{split}
    \end{equation*}
    Thus, the maximand in Equation \ref{maximand} has decreasing differences.
\end{proof}

\subsection*{PROOF OF PROPOSITION \ref{extreme}}

\begin{proof}
    For exposure to be chosen at all addictive states, it must be that
    \begin{equation*}
        \begin{split}
        &\max_{q \in [0,Q]}\Big\{(1-\mcal P(s))[u(s,q) + \delta V(\lambda s + q, Q)] \\ 
        &\hspace{3em}+\mcal P(s)[u(s,Q) + \delta V(\lambda s + Q, Q)]\Big\} \\
        > & u(s, y(s)-r(s), 0, R) + \delta V(\lambda s, 0)
        \end{split}
    \end{equation*}
    at all $s$ for some $Q$. Then, for any $Q$
    \begin{equation*}
        \begin{split}
        &\max_{q \in [0,Q]}\Big\{(1-\mcal P(s))[u(s,q) + \delta V(\lambda s + q, Q)] \\ 
        &\hspace{3em}+\mcal P(s)[u(s,Q) + \delta V(\lambda s + Q, Q)]\Big\} \\
        > & u(s,Q) + \delta V(\lambda s + Q, Q)
        \end{split}
    \end{equation*}
    so it suffice to show that there exists $Q$ such that
    $$u(s,Q) + \delta V(\lambda s + Q, Q) > u(s, y(s)-r(s), 0, R) + \delta V(\lambda s, 0)$$
    for all $s$.

    Now, suppose that the quantity limit is zero. Then,
    $$u(s, y(s), 0, E) + \delta V(\lambda s, 0) > u(s, y(s)-r(s), 0, R) + \delta V(\lambda s, 0)$$
    for all $s$. By continuity of $u$ in $s$ and the fact that the set of all possible addictive states $[0, \overline{Q}/(1-\lambda]$ is closed, we have that
    $$\min_s \left\{u(s, y(s), 0, E) - u(s, y(s)-r(s), 0, R)\right\} = \epsilon > 0.$$
    As such, it suffice to find $Q$ such that 
    \begin{equation}\label{cond}
        u(s,Q) + \delta V(\lambda s + Q, Q) > u(s,y(s),0, E) + \delta V(\lambda s, 0) - \epsilon.
    \end{equation}
    Let 
    $$h(s,Q) = u(s,y(s),0, E) - u(s,Q) + \delta(V(\lambda S, 0) - V(\lambda s + Q, Q)).$$
    Then, Equation \ref{cond} is equivalent to 
    $$\epsilon > h(Q,s).$$
    By continuity of $u$ and $V$, we have that $h$ is continuous in $s, Q$. As such, $h^{-1}([0,\epsilon))$ is an open set (where $[0, \epsilon)$ is the half-open interval from $0$ to $\epsilon$) in $[0, \overline{Q}/(1-\lambda] \times [0, \overline{Q}]$. As such, the complement of $h^{-1}((0,\epsilon))$ is closed.

    Towards a contradiction, suppose there does not exist $Q > 0$ such that $(s, Q) \in h^{-1}([0,\epsilon))$ for all $s$. Then, there must exist a sequence $\{s_k, Q_k\}_{k = 1}^\infty$ such that $(s_k, Q_k) \in (h^{-1}([0,\epsilon)))^C$ for all $k$ and $\lim_{k \to \infty} Q_k = 0$. Then, letting $\hat{s} = \lim_{k \to \infty} s_k$, it must be that $(\hat{s},0) \in (h^{-1}([0,\epsilon)))^C$. However, $h(\hat{s},0) = 0$, a contradiction.
\end{proof}

\end{document}